\documentclass[10pt, conference]{IEEEtran}
\usepackage{blindtext, graphicx}
\usepackage{fullpage}
\usepackage{graphicx}
\usepackage{amsmath}
\usepackage{amsthm}
\usepackage{cite}
\usepackage{amsfonts}

\IEEEoverridecommandlockouts
\begin{document}
\title{Cascading Node Failure with Continuous States in Random Geometric Networks}

\author{\IEEEauthorblockN{Khashayar Kamran\IEEEauthorrefmark{1} and
Edmund Yeh\IEEEauthorrefmark{2}}
\IEEEauthorblockA{Department of Electrical and Computer Engineering\\
Northeastern University\\
Email: \IEEEauthorrefmark{1}kamran.k@husky.neu.edu,
\IEEEauthorrefmark{2}eyeh@ece.neu.edu}}

\maketitle

\allowdisplaybreaks

\newtheorem{Thm}{\textbf{Theorem}}
\newtheorem{Lem}{\textbf{Lemma}}
\newtheorem{Cor}{\textbf{Corollary}}
\newtheorem{Def}{\textbf{Definition}}
\newtheorem{Exam}{\textbf{Example}}
\newtheorem{Alg}{\textbf{Algorithm}}
\newtheorem{Sch}{\textbf{Scheme}}
\newtheorem{Prob}{\textbf{Problem}}
\newtheorem{Rem}{\textbf{Remark}}
\newtheorem{Prp}{\textbf{Proposition}}
\newtheorem{Asump}{\textbf{Assumption}}
\newtheorem{Subp}{\textbf{Subproblem}}

\begin{abstract}
The increasing complexity and interdependency of today's networks highlight the importance of studying network robustness to failure and attacks. 
Many large-scale networks are prone to cascading effects where a limited number of initial failures (due to attacks, natural hazards or resource depletion) propagate through a dependent mechanism, ultimately leading to a global failure scenario where a substantial fraction of the network loses its functionality.  These cascading failure scenarios often take place in networks which are embedded in space and constrained by geometry.  Building on previous results on cascading failure in random geometric networks, we introduce and analyze a continuous cascading failure model where a node has an initial continuously-valued state, and fails if the aggregate state of its neighbors fall below a threshold.  Within this model, we derive analytical conditions for the occurrence and non-occurrence of cascading node failure, respectively.   
\end{abstract}


\section{Introduction}

Many large-scale networks are prone to cascading effects where a limited number of initial node and/or link failures (due to attacks, natural hazards or resource depletion) lead to further failures through a dependent mechanism, ultimately leading to a global failure scenario where a substantial fraction of the network nodes and/or links lose their functionality.  Such cascading effects can be observed in electrical power networks, communication networks, social networks, and biological networks.  In electrical power networks, the failure of a few power lines carrying load can result in redistribution of the load onto other lines, which may fail due to excessive load, thus potentially leading to a cascading power blackout~\cite{P7}.  In communication networks, a similar phenomenon may occur as the failure of a few communication nodes and/or links causes traffic to be redirected to other links, causing additional congestion, which may snowball into a global state of ``congestion collapse."  In social and economic networks, initial adoption of opinions or decisions by a limited number of individuals may spread to others through social influencing mechanisms, resulting in widespread adoption on a global scale~\cite{P2}. Cascading effects can be observed not only in single networks, but also in interdependent networks where the failure of a small number of nodes in just one network lead to failures in the other network, and vice versa, resulting in a significant part of the interdependent network becoming unavailable.  A classical example of this is the interdependence between the power grid and the communication infrastructure, where power failures can lead to communication failures, which in turn, lead to cascading power failures~\cite{P3}.  

Three salient features arise from the cascading failure scenarios discussed above.  First, the state of a node or link typically depends on the states of multiple neighbors.  This is true not only in single networks but also in interdependent networks, where for instance, a node in the communication network may receive power from several substations in the power network and a power substation may be controlled by several communication nodes.   Second, in predicting the occurrence or non-occurrence of cascading failure, it is often necessary to describe the functional state of a node or link in more refined terms than simply binary.  For example, a communication node receiving power from several power substations may fail only if the aggregate power delivered by the power substations is below a required threshold.   Likewise, a power node controlled by several communication nodes may fail only if the aggregate communication bandwidth from the communication nodes falls below a threshold.  In these cases, the state of a functioning node is more usefully described by a continuously-valued parameter (e.g. power supplied or bandwidth offered).   Third, in assessing cascades in networks spatial constraints are often important considerations.  Real networks such as power grids, communication networks, and biological networks are embedded in space, highlighting the role of geometry.  For instance, in wireless communication networks, worms and viruses planted in mobile devices can infect geographically nearby mobile devices, potentially propagating a ``wireless epidemic" in the manner of biological epidemics~\cite{P4}.  The importance of geometry extends to interdependent networks.  In the scenario observed in Italy in September 2003, where an iterative cascade took place between power and communication networks, every internet server in the communication network is connected to the
geographically closest power station~\cite{P3}.

To address the key issues discussed above, we present and analyze a new continuous model for cascading node failure in large-scale networks.  In this model, each node has an initial state represented by a continuously-valued parameter.  This parameter reflects, for instance, the resource level (e.g. power or bandwidth) that the node can provide to its neighboring nodes.   We examine failure scenarios which start with an initial node failure (resulting from attacks, natural hazards or resource depletion).  When a node fails, its state falls to zero.  The dependent mechanism by which failures propagate in the network is described as follows.   Each node has a continuously-valued threshold which describes its requirement for resources from its neighboring nodes.  A node fails if and only if it has at least one failed neighbor and the sum of the states of its neighboring nodes falls below its threshold.   Such a failure mechanism is applicable to both single networks and interdependent networks.  In this paper, we focus on cascading failure scenarios under the given mechanism in the context of single networks.  To reflect important spatial constraints, we shall examine cascading node failure within the context of random geometric graphs (RGG), in which nodes are spatially distributed according to Poisson point process of a given density and two nodes share a link if their distance is below a given distance parameter.  We will develop analytical conditions for the occurrence or non-occurrence of cascading failure under the continuous failure model for random geometric networks.   

Cascading failure in random geometric networks was first studied in \cite{P5}, where every node has binary state (either functional or failed) and fails if the fraction of its neighboring nodes which have failed exceeds a given threshold.   A degree-dependent percolation model is used to find analytical conditions for the occurrence of non-occurrence of cascading node failure~\cite{P5}.  Cascading link failure in random geometric networks is analyzed in~\cite{P6}, which uses a mapping from links in a RGG to nodes in corresponding dual covering graph to derive analytical results on degree-dependent and cascading link failure.  The continuous cascading failure model presented in this paper can be seen as a generalization of the degree-dependent models presented in~\cite{P5, P6}.  Note that while we focus on cascading node failures in this paper, the results may be applied to cascading link failure scenarios using the techniques of~\cite{P6}. 

Cascading failure in networks where nodes have continuously-valued state have been studied in~\cite{P7, P8}.  In these papers, the state of a node is defined as the load it carries.  When a node failed, its load is added to the loads of other nodes.  For instance, in~\cite{P8}, loads are transmitted between pairs of nodes through the shortest path between them.  When a node fails along with its adjacent links, shortest paths must change and the load on the new shortest paths increase.  When the load on a node on the new shortest paths exceeds a given capacity, the node fails and this failure may propagate, leading to cascading failure in the network.  The continuous cascading failure model presented in this paper considers a different mechanism from those in~\cite{P7, P8}.  Furthermore, we consider practical spatial constraints through modeling of random geometric graphs. 

We present the continuous failure model in Section II, identify the special role played by ``highly vulnerable" and ``highly reliable" nodes, and derive analytical conditions for the occurrence and non-occurrence of cascading node failure (respectively) in RGGs under the continuous failure model.  We then evaluate the conclusions of the analysis via numerical simulations in Section III.  

\section{Network and Continuous Failure Model}
\subsection{Network Model}

Consider a large-scale network modeled by an (infinite) RGG \(G(\mathcal{H}_\lambda,r)\), where the vertex set \(\mathcal{H}_\lambda\) is a homogeneous Poisson point process in \(\mathbb{R}^2\) with density \(\lambda > 0\), and two nodes \(X_i,X_j \in \mathcal{H}_\lambda\) share an undirected link if and only if \( ||X_i - X_j|| \leq r~(r > 0)\), with  \( ||X_i - X_j||\) being the Euclidean distance between \(X_i\) and \(X_j\).  Since RGGs have the scaling property \cite{P10, P11}, we focus on \( G(\mathcal{H}_\lambda,1)\) in the following. 

Let $\mathcal{H}_{\lambda, {\bf 0}} =  \mathcal{H}_{\lambda} \cup \{\bf 0\}$ be the union of the origin and $\mathcal{H}_{\lambda}$. It is well known from the study of continuum percolation that there exists a critical density $\lambda_c$, $0 < \lambda_c < \infty$, such that for $\lambda < \lambda_c$, the component containing the origin in the graph $G(\mathcal{H}_{\lambda, {\bf 0}},1)$ has a finite number of nodes almost surely, while for $\lambda > \lambda_c$, the component containing the origin has an infinite number of nodes with positive probability.  Furthermore it is known that if $\lambda > \lambda_c$, there exists a {\em unique} infinite component in \( G( \mathcal{H}_\lambda,1)\) \cite{P9, P10, P11, P12}. 

For analytical convenience, we consider RGG in the infinite graph setting.  One can analogously define a finite RGG $G({\mathcal{X}	}_n, 1)$ with $n$ nodes.  It can be shown that if $\lambda > \lambda_c$, there exists a unique largest component of size $\Theta (n)$ in $G({\cal X}_n, 1)$, called the giant component (GC). If $\lambda < \lambda_c$, the largest component in $G({\mathcal{X}}_n, 1)$ can have only size $O(\log n)$~\cite{P11}. 

\subsection{Continuous Failure Model}

We now investigate the cascading failure phenomenon in random geometric networks under a continuous failure model.  Assume all nodes in $G(\mathcal{H}_{\lambda},1)$, where $\lambda > \lambda_c$, are initially functional.   Now consider an initial failure seed represented by a single failed node.   The question is whether this initial failure can lead, via the dependent failure mechanism to be discussed below, to a global cascade of failures in $G(\mathcal{H}_{\lambda},1)$.  As discussed in~\cite{P5}, a cascading failure can be characterized by establishing whether the network has been affected in a global manner rather than in an isolated local manner.  This is a notion which can be appropriately captured using the concept of percolation.  
\begin{Def}~\cite{P5}
A cascading failure is an ordered sequence of node failures triggered by an initial failure, resulting in an infinite component of failed nodes in the network.
\end{Def}

For each node $j  \in \mathcal{H}_\lambda$, associate a continuously-valued state $S_j \in [0,1]$ representing, for instance, the resource level that node $j$ can provide to its neighboring nodes.  Also associate with each node $j$ a continuously-valued threshold $\phi_j \in [0,\infty)$ representing the resource level required by node $j$ to stay functional.  We shall assume that the initial node states $S'_j$ (before any failures) are i.i.d. random variables taking values in $(0,1]$ with probability density function $f_{S'}$ and that the $\phi_j$'s are i.i.d. random variables with probability density $f_\phi$.  Furthermore, assume that $\{S'_j\}$ and and  $\{\phi_j\}$ are mutually independent.

\textbf{Node Failure Model:}  If node \(j \) has at least one failed neighbor, and \( \sum_{k \in \mathcal{K}_j}  S_k < \phi_j\), where \(\mathcal{K}_j\) is the set of neighbors of \(j\), then node \(j\) fails.  The state of node $j$ is zero $(S_j = 0)$ if and only if $j$ is a failed node.  Otherwise, node $j$ is said to be functional.

\subsection{Analysis of Cascading Failure}

In~\cite{P5}, it is shown that the set of {\em vulnerable} nodes and the set of {\em reliable} nodes play key roles in determining the occurrence or non-occurrence of cascading node failure in RGGs with a binary degree-dependent failure model.  In the context of~\cite{P5}, a vulnerable node is a node which fails if just one of its neighbors fails.  A reliable node is a node which does not fail unless all its neighbors fail.  An {\em unreliable} node is a node which is not reliable.

It turns out that similar notions are key for determining the occurrence or non-occurrence of cascading node failure within the proposed continuous-state failure model. 
\begin{Def}
A {\em highly vulnerable} node is a node whose threshold satisfies
\begin{small}
\begin{equation*}
\phi > k - \displaystyle\min_{i \in \mathcal{K}} \{S_i'\},
\end{equation*}
\end{small}
where $\mathcal{K}$ is the set of the neighbors of the node, and $k = |\mathcal{K}|$.

\end{Def}
\begin{Def}
A {\em highly reliable} node is a node whose threshold satisfies
\begin{small}
\begin{equation*}
\phi \leq \displaystyle\min_{i \in \mathcal{K}} \{S'_i\}.
\end{equation*}
\end{small}
A node which is not highly reliable is said to be a {\em weak} node.
\end{Def}

From the above definitions, it is clear that a highly vulnerable node fails if any one of its neighbors fails.  Thus, the set of highly vulnerable nodes is a subset of the vulnerable nodes.
\footnote{Note that \(\phi > k - 1\) is also sufficient for a node to be vulnerable.  Indeed, this may seem to be a tighter condition than the condition given in Definition~2.  It turns out, however, that \(\phi > k - 1\) does not lead to a useful sufficient condition for a GC of vulnerable nodes.  One reason is that the condition does not involve the distribution of $S'$.}
Similarly, a highly reliable node remains functional as long as one of its neighbors is functional.  Thus, the set of highly reliable nodes is a subset of the reliable nodes.

\begin{Lem}
The probability of a node being highly vulnerable is
\begin{small}
\begin{equation*}
\rho_k = \bigg[1-\int _{0} ^{k}F_{S'}(k-\varphi)f_\phi(\varphi)d\varphi\bigg]^k,
\end{equation*}
\end{small}
where $k$ is the number of neighbors of the node, $F_{S'}$ is the distribution function 
of the initial state $S'$, and $f_\phi$ is the probability density function of the threshold
$\phi$.  The probability of a node being highly reliable is
\begin{small}
\begin{equation*}
\sigma_k = \bigg[1 - \int _{0} ^{\infty}F_{S'}(\varphi)f_\phi(\varphi)d\varphi \bigg]^k.
\end{equation*}
\end{small}
\end{Lem}

\begin{proof}  We have
\begin{small}
\begin{align*}
\rho_k &= P_r \{\phi > k - \displaystyle\min_{i \in K} \{S_i'\}\} = P_r \{\displaystyle\min_{i \in K} \{S_i'\} >  k-\phi\} \\
&\overset{(a)}{=} P_r \{S' >  k-\phi\} ^k \\
&\overset{(b)}{=}   \bigg[ \int _{0} ^{\infty}P_r \{S' > k-\phi | \phi = \varphi\} f_\phi(\varphi)d\varphi  \bigg]^k \\
 &= \bigg[ \int _{0} ^{\infty}(1 - P_r \{S' \leq k-\varphi\}) f_\phi(\varphi)d\varphi  \bigg]^k \\
 &=  \bigg[ \int _{0} ^{\infty}f_\phi(\varphi)d\varphi     -     \int _{0} ^{\infty}F_{S'}(k-\varphi)f_\phi(\varphi)d\varphi  \bigg]^k \\
 &=  \bigg[1-\int _{0} ^{\infty}F_{S'}(k-\varphi)f_\phi(\varphi)d\varphi \bigg]^k \\
&= \bigg[1-\int _{0} ^{k}F_{S'}(k-\varphi)f_\phi(\varphi)d\varphi\bigg]^k\\
\end{align*}
\end{small}
Here, (a) is due to the mutual independence of  the \(S'_i\)'s and (b) is due to the mutual independence of the \(S'_i\)'s and \(\phi_i\)'s.  Next, we have
\begin{small}
\begin{align*}
\sigma_k &= P_r \{\phi \leq \displaystyle\min_{i \in K} \{S'_i\}\} = P_r\{ S' \geq \phi \}^k \\
&= \bigg[ \int _{0} ^{\infty}P_r\{S' \geq \phi | \phi = \varphi\}f_\phi(\varphi)d\varphi \bigg]^k \\
&=  \bigg[ \int _{0} ^{\infty}(1-P_r\{S' < \varphi\})f_\phi(\varphi) d\varphi \bigg]^k \\
&= \bigg[\int _{0} ^{\infty}f_\phi(\varphi)d\varphi -     \int _{0} ^{\infty}F_{S'}(\varphi)f_\phi(\varphi)d\varphi \bigg]^k \\
&= \bigg[1 - \int _{0} ^{\infty}F_{S'}(\varphi)f_\phi(\varphi)d\varphi \bigg]^k.
\end{align*}
\end{small}
\end{proof}

Note that the probabilities $\rho_k$ and $\sigma_k$ depend only on the degree $k$ and the distributions of $S'$ and $\phi$.  

We will now use the concepts of highly vulnerable and highly reliable nodes to derive a sufficient condition for the occurrence of cascading node failure and a sufficient condition for the non-occurrence of cascading node failure, respectively.





\begin{Thm}
For any \(\lambda_1 > \lambda_c\) and \(G(\mathcal{H}_\lambda,1)\) with \(\lambda > \lambda_1\), there exists \(k_0 < \infty\) depending on \( \lambda \) and \(\lambda_1 \) such that if
\begin{small}
\begin{equation*}
\bigg[1-\int _{0} ^{k}F_{S'}(k-\varphi)f_\phi(\varphi)d\varphi \bigg]^{k} \geq \frac{\lambda_1}{\lambda}
\end{equation*}
\end{small}
for all \(1 \leq k \leq k_0\),
then with probability 1, there exists an infinite component of highly vulnerable nodes in \(G(\mathcal{H}_\lambda,1)\).  Moreover, if the initial failure is inside this component or adjacent to some node in this component, then with probability 1, there is cascading node failure in  \(G(\mathcal{H}_\lambda,1)\).
\end{Thm}

\begin{proof}
We view the problem as a degree-dependent node failure problem where a highly vulnerable node is considered ``operational" and all other nodes are considered ``failed".  In this model, a node with degree $k$ fails with probability $1-\rho_k$ where $\rho_k$ is given by Lemma 1.  Then, by applying Theorem 1-(i) from~\cite{P5}, we obtain the sufficient condition for having an infinite component of highly vulnerable nodes in \(G(\mathcal{H}_\lambda,1)\).  Furthermore, since highly vulnerable nodes are also vulnerable, it is clear that if the initial failure is inside this component or adjacent to some node in this component, then there is cascading node failure in  \(G(\mathcal{H}_\lambda,1)\).
\end{proof}
\begin{Thm}
 For any \( G({\cal H}_\lambda,1)\) with \(\lambda > \lambda_c\), if
 \begin{small}
 \begin{equation*}
\displaystyle \sum_{k = 0} ^{\infty} \frac{(\frac{\lambda}{2})^k}{k!} e^{-(\frac{\lambda}{2})} \bigg( \displaystyle \sum_{m = 0} ^{\infty} \frac{[\lambda(2\sqrt2 + \pi)]^m}{m!} e^{-\lambda(2\sqrt2 + \pi)} 
\end{equation*}
\end{small}
\begin{small}
\begin{equation}
\times \bigg(1 - \bigg[1 - \int _{0} ^{\infty}F_{s'}(\varphi)f_\phi(\varphi)d\varphi \bigg]^{(m+k-1)k} \bigg) \bigg)< \frac{1}{27}
\label{eq:nocascade}
\end{equation}
\end{small}
Then with probability 1, there is no infinite component of weak nodes. As a consequence, with probability 1, there is no cascading
failure in \( G(\mathcal{H}_\lambda,1)\) no matter where the initial failure is. 
\end{Thm}

\begin{proof}
We view the problem as a degree-dependent node failure problem where weak nodes are considered ``operational" and highly reliable nodes are ``failed."  In this model, a node with degree $k$ fails with probability $\sigma_k$ where $\sigma_k$ is given by Lemma 1.  Note that $\sigma_k$ is non-increasing in $k$.  Then, by Theorem 1-(ii) of~\cite{P5}, when~\eqref{eq:nocascade} holds, with probability 1, there is no infinite component of weak nodes.  Since the set of highly reliable nodes is a subset of the reliable nodes, the set of weak nodes is a superset of the unreliable nodes.  Thus, when~\eqref{eq:nocascade} holds, with probability 1, there is no infinite component of unreliable nodes.  By Theorem 2(ii) of~\cite{P5}, this implies that there is no cascading failure in  \( G(\mathcal{H}_\lambda,1)\) no matter where the initial failure is. 
\end{proof}

We now analytically evaluate the above conditions for the occurrence and non-occurrence of cascading failure within the context of an illustrative example. 

\begin{Exam}

Consider \(G(\mathcal{H}_\lambda,1)\) with \(\lambda > \lambda_c\).   Assume the initial node states $\{ S_j\}$ are i.i.d. uniform random variables on $(0,1]$.  Furthermore, assume the node thresholds $\{ \phi_j\}$ are i.i.d. exponential random variables with parameter $\mu > 0$ (mutually independent of the $S_j$'s).   

\end{Exam}

We first use Theorem 1 to find a sufficient condition for cascading failure for Example 1.   Thus, we seek the existence of a $k_0$ such that 
\begin{small}
\begin{equation*}
\bigg[1-\int _{0} ^{k}F_{S'}(k-\varphi)f_\phi(\varphi)d\varphi \bigg]^{k} \geq \frac{\lambda_1}{\lambda}
\end{equation*}
\end{small}
for all \(1 \leq k \leq k_0\). Since \(k \geq 1\), we have
\begin{small}
\begin{align*}
& \bigg[1-\int _{0} ^{k}F_{S'}(k-\varphi)f_\phi(\varphi)d\varphi \bigg]^{k} \\
= & \bigg[1-\int_{0}^{k-1}F_{S'}(k-\varphi)f_\phi(\varphi)d\varphi \\
 - & \int_{k-1}^{k}F_{S'}(k-\varphi)f_\phi(\varphi)d\varphi \bigg]^k \\
= & \bigg[1-\int _{0} ^{\infty}\mu e^{-\mu \varphi}d\varphi  - \int_{k-1}^{k}(k-\varphi)\mu e^{-\mu\varphi}d\varphi \bigg]^k \\
= & \left(\frac{e^{\mu}-1}{\mu}\right)^ke^{-\mu k^2}.
\end{align*}
\end{small}

This leads to the following simple and intuitive condition for having an infinite component of highly vulnerable nodes within the setting of Example 1, which in turn implies the existence of cascading failure if the initial failure is inside this component or adjacent to some node in this component. 
\begin{Prp}
Within the setting of Example 1, if \(\mu < m'(\frac{\lambda_1}{\lambda})\), where \( m'(\frac{\lambda_1}{\lambda})\) is the solution to the equation \((\frac{1 - e^{-x}}{x}) = \frac{\lambda_1}{\lambda}\), then an infinite component of highly vulnerable nodes exists in \(G(\mathcal{H}_\lambda,1)\).
\end{Prp}

\begin{proof} 
Since \((\frac{1 - e^{-x}}{x})\) is monotonically decreasing in \(x\), if  \(\mu < m'(\frac{\lambda_1}{\lambda})\), we have \( \frac{1 - e^{-\mu}}{\mu} > \frac{\lambda_1}{\lambda}\).  Thus, \( (\frac{e^{\mu}-1}{\mu})^ke^{-\mu k^2} > \frac{\lambda_1}{\lambda}\) for \(k =1\). Since \( (\frac{e^{\mu}-1}{\mu})^ke^{-\mu k^2}\) is decreasing and continuous in \(k \geq 1\) for \(\mu > 0 \), there exists \(k_0 \geq 1\) such that \((\frac{e^{\mu}-1}{\mu})^ke^{-\mu k^2} \geq \frac{\lambda_1}{\lambda}\) for all \(1 \leq k \leq k_0\).  Thus, \(\mu < m'(\frac{\lambda_1}{\lambda})\) is a sufficient condition for having an infinite component of highly vulnerable nodes. 
\end{proof}

Note that the result of Proposition 1 is intuitively reasonable.  When $\mu$ is sufficiently small, the probability of a node having a high threshold is large, and therefore cascading failure becomes likely.   

Next, we use Theorem 2 to find a sufficient condition for not having an infinite component of weak nodes, and therefore not having cascading failure, in the setting of Example 1.  For this, we fix the density to be \(\lambda = 4\).\footnote{Note that the best known analytical upper bound on \(\lambda_c\) is 3.372.}   The sufficient condition reduces to
\begin{small}
\begin{equation*}
\displaystyle \sum_{k = 0} ^{\infty} \frac{(\frac{4}{2})^k}{k!} e^{-(\frac{4}{2})} \bigg( \displaystyle \sum_{m = 0} ^{\infty} \frac{[4(2\sqrt2 + \pi)]^m}{m!} e^{-4(2\sqrt2 + \pi)} 
\end{equation*}
\begin{equation*}
\times \bigg(1 - \left[1 - \int _{0} ^{\infty}F_{S'}(\varphi)f_\phi(\varphi)d\varphi\right]^{(m+k-1)k} \bigg) \bigg) = 
\end{equation*}

\begin{equation*}
\displaystyle \sum_{k = 0} ^{\infty} \frac{(\frac{4}{2})^k}{k!} e^{-(\frac{4}{2})} \bigg( \displaystyle \sum_{m = 0} ^{\infty} \frac{[4(2\sqrt2 + \pi)]^m}{m!} e^{-4(2\sqrt2 + \pi)} 
\end{equation*}

\begin{equation}
\times \bigg(1 - \left[1 + \frac{e^{-\mu} -1}{\mu} \right]^{(m+k-1)k} \bigg) \bigg) < \frac{1}{27}.
\end{equation}
\end{small}


Since \(1 - (1 + \frac{e^{-\mu} -1}{\mu})^{(m+k-1)k}\) is decreasing in \(\mu\), the LHS of (1) is decreasing in \(\mu\).  The LHS of (2) is smaller than \(\frac{1}{27}\) for \(\mu \geq 1357\) \( (\lambda = 4) \).  Thus for \(\mu \geq 1357\), there is no infinite component of weak nodes, and therefore no cascading failure.  Again, this result is intuitively reasonable.  When $\mu$ is sufficiently large, the probability of a node having a high threshold is small, and therefore cascading failure becomes unlikely.    We will evaluate the above conditions on $\mu$ through numerical simulation in Section~\ref{sec:simulations}. 


\section{Numerical experiments}

\label{sec:simulations}


In this section, we experimentally evaluate Example~1.  The initial network is a finite RGG $G({\cal X}_n, 1)$ with density $\lambda = 4$ and $n= 1600$ nodes independently and uniformly distributed in $[0, 20]^2$.  Initial node states are i.i.d uniformly distributed on (0,1] and node thresholds are i.i.d exponentially distributed with parameter \(\mu\).  Figures 1 and 2 illustrate the sufficient condition for the existence of a GC of highly vulnerable nodes and thus cascading node failure. 
We choose $\lambda_c < \lambda_1 = 3.85 < \lambda$. For these values, according to Proposition 1, $\mu < 0.076$ is a sufficient condition for the existence of a GC of highly vulnerable nodes. Figure 1 shows the presence of GC of highly vulnerable nodes (represented by solid nodes) when \(\mu = 0.075\). Figure 2 depicts the cascading failure after the initial failure with \(\mu = 0.075\), where empty nodes represent failed nodes and the location of the initial failure is shown by an arrow.

Figure 3 and 4 illustrate the sufficient condition for not having a GC of weak nodes.  The condition for Example~1 states that if \(\mu > 1357\), there is no GC of weak nodes in the network.  Figure 3 shows the same network with \(\mu = 1360\) where empty nodes represent the weak nodes and solid nodes represent highly reliable nodes. Figure 4 shows the described network after the initial failure where solid nodes are functional and empty nodes are failed.  As expected from a network with no GC of weak nodes, there is no cascading failure.

\begin{figure}[!htb]
\minipage{0.45\textwidth}
  \includegraphics[width=\linewidth]{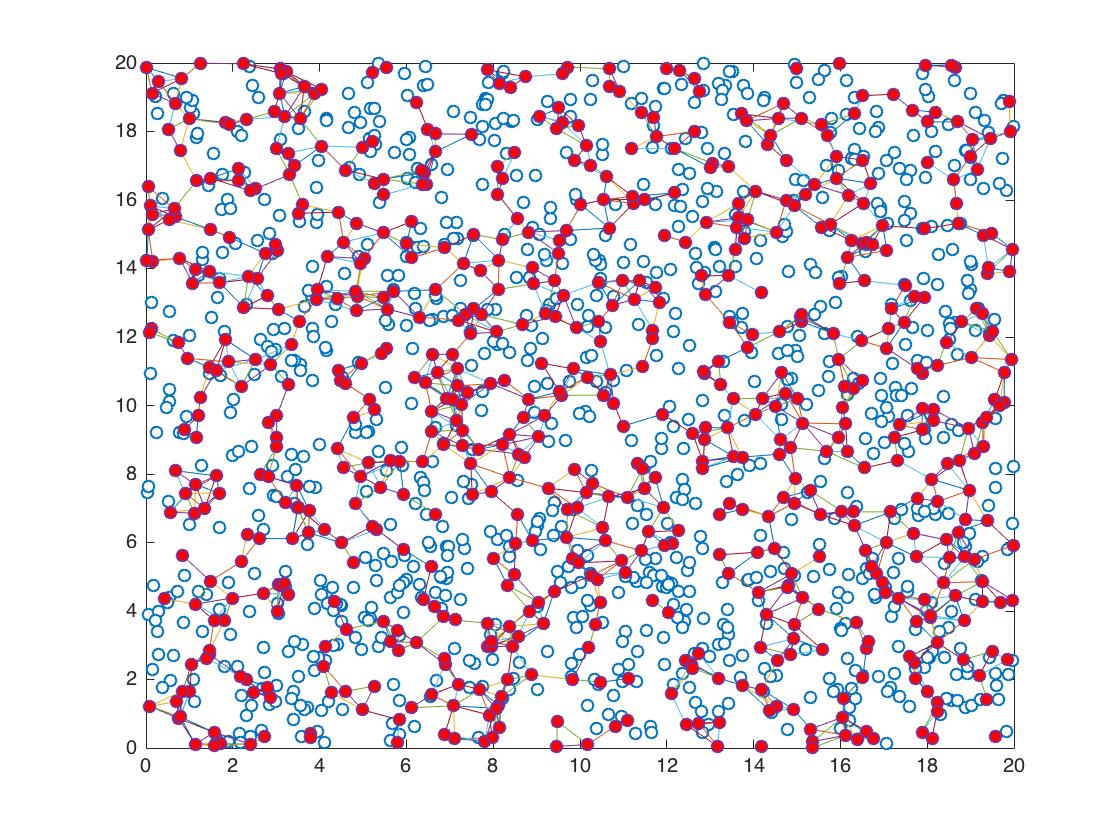}
  \caption{\(\mu = 0.075 \) ; GC of highly vulnerable nodes (solid)}
\endminipage\hfill
\minipage{0.45\textwidth}
  \includegraphics[width=\linewidth]{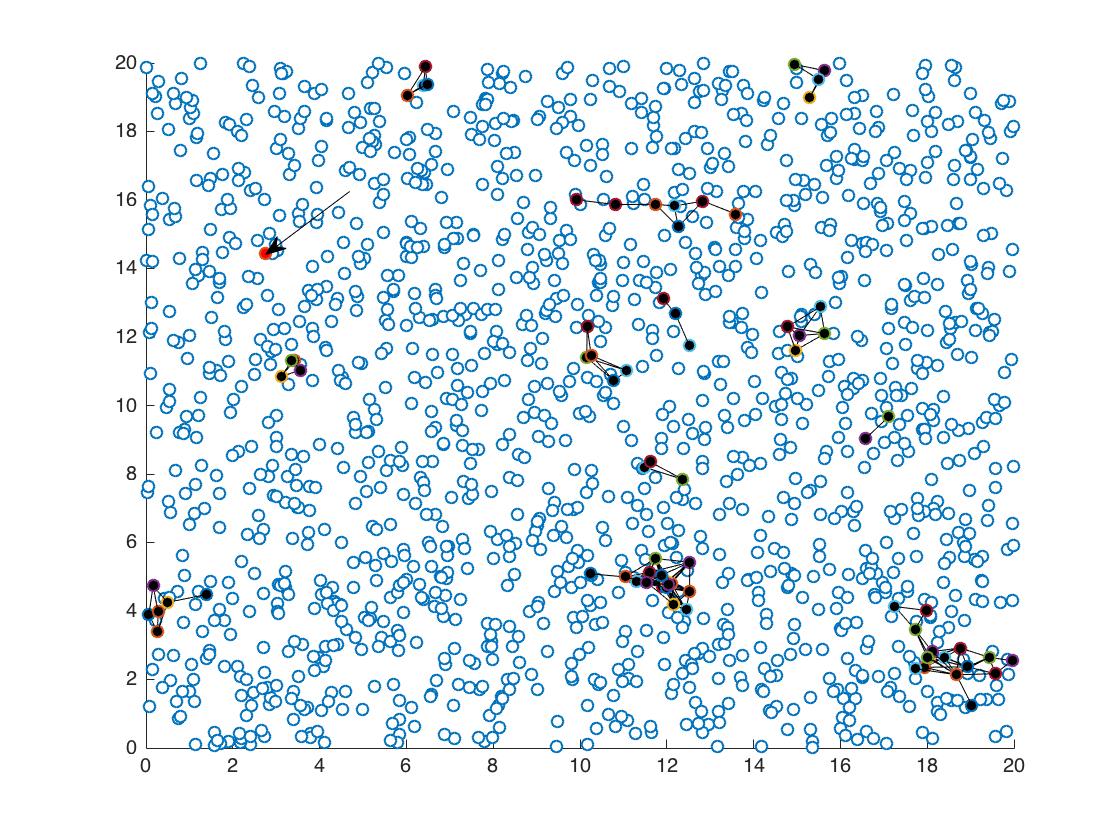}
  \caption{\(\mu = 0.075 \) ; cascading failure ; initial failure indicated by arrow; functional nodes are solid}
\endminipage
\end{figure}



\begin{figure}[!htb]
\minipage{0.45\textwidth}
    \includegraphics[width=\linewidth]{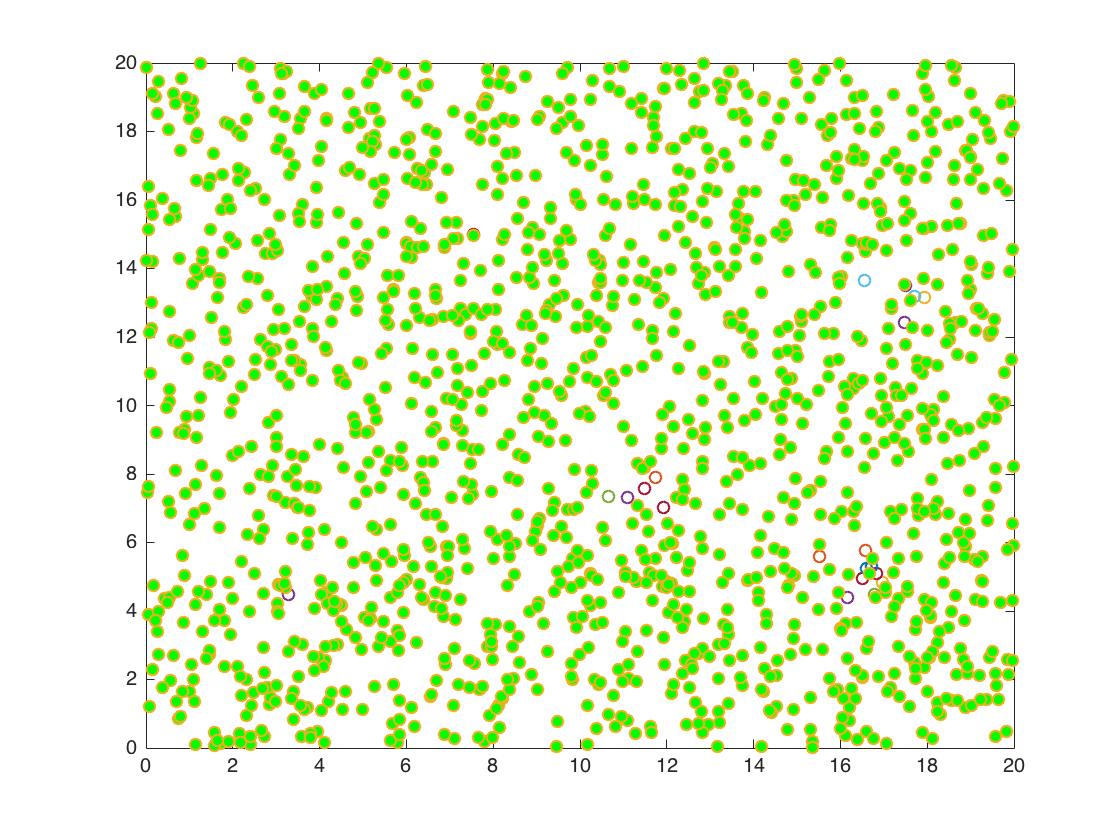}
  \caption{\(\mu = 1360 \) ; no GC of weak nodes (empty)}
\endminipage\hfill
\minipage{0.45\textwidth}
    \includegraphics[width=\linewidth]{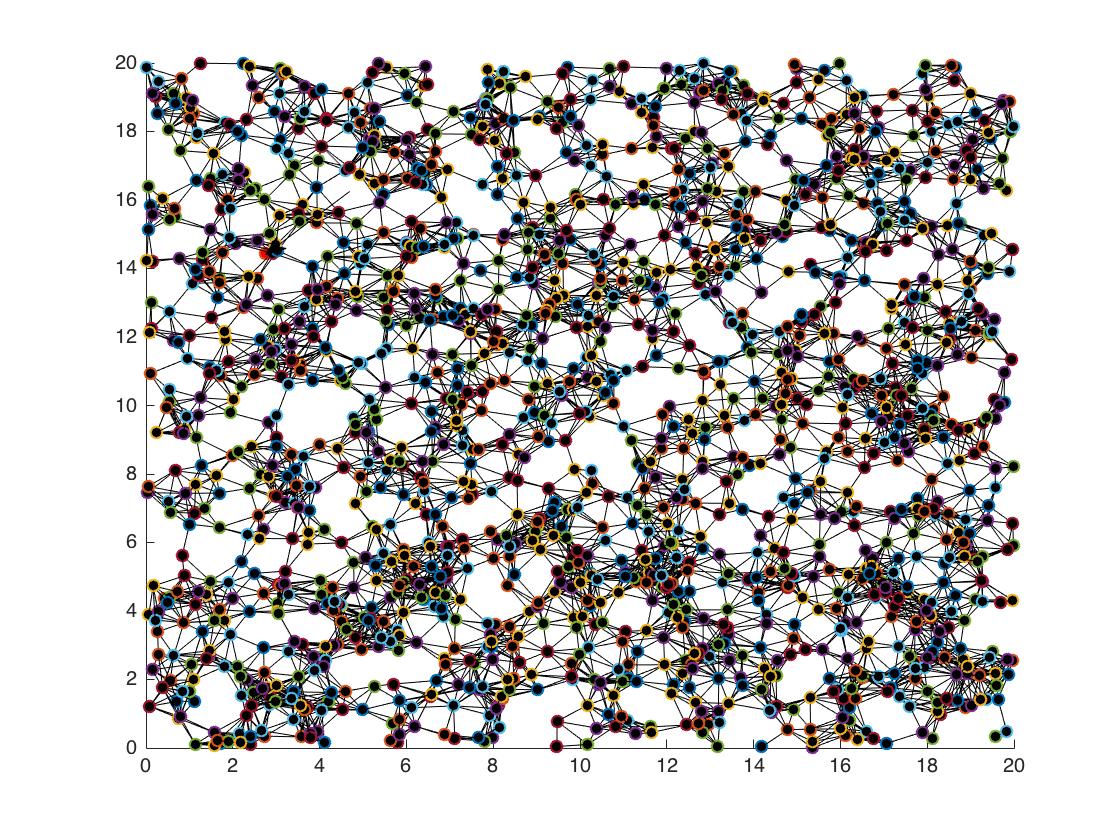}
  \caption{\(\mu = 1360 \) ; no cascading failure ; functional nodes are solid}
\endminipage
\end{figure}

\section{Conclusion}
We introduced and analyzed a model for cascading node failure in random geometric networks where nodes have a continuously valued (rather than simply binary) state.  The concepts of highly vulnerable and highly reliable nodes were introduced to obtain an analytical necessary condition and an analytical sufficient condition for the occurrence and non-occurrence of cascading failure, respectively.  These analytical conditions were experimentally verified through numerical simulation.  Although we focused on a single geometric network, the continuous model is also useful for analyzing cascading failure in interdependent geometric networks where nodes in one network provide resources (e.g. power or bandwidth) to nodes in the other network.  The application of this model to the interdependent case will be treated in future work.

\bibliography{Globecom2016Cascade_vf_EY}{}
\bibliographystyle{unsrt}
\end{document}